\documentclass[%
aip,
 jmp,%
 amsmath,amssymb,
preprint,%
author-numerical,%
]{revtex4-1}
\usepackage{bm}
\usepackage{braket} 
\usepackage{epsfig}
\usepackage{amsthm}
\usepackage{amssymb}
\usepackage{latexsym}
\usepackage{url}
\usepackage{amsmath}
\usepackage{verbatim}
\usepackage[utf8]{inputenc}
\usepackage[english]{babel}

\newtheorem{thm}{Theorem}

\theoremstyle{definition}

\newtheorem*{remark}{Remark}

\newcommand{\erre}{\mathbb{R}}
\newcommand{\tr}{\text{tr}}
\newcommand{\Mnn}{\mathcal{M}_{n,n} (\mathbb{C})}
\newcommand{\Pm}{\mathcal{P}^+}
\newcommand{\Vol}{\operatorname{Vol}}

\usepackage{graphicx}
\usepackage{dcolumn}
\usepackage{bm}

\begin{document}

\preprint{AIP/123-QED}

\title[Geometric approach to Quantum Dynamics]{Geometric approach to non-relativistic Quantum Dynamics of mixed states}

\author{Vicent Gimeno}
\email{vigigar@postal.uv.es}
 \affiliation{Departament de Matem\`{a}tiques- Institute of New Imaging Technologies, Universitat Jaume I, Castell\'o,
Spain.}

\author{Jose M. Sotoca}%
 \email{sotoca@uji.es}
\affiliation{ Departamento de Lenguajes y Sistemas Inform\'aticos- Institute of New Imaging Technologies, Universitat Jaume I, Castell\'o,
Spain.
}%

\date{\today}

\begin{abstract}
 In this paper we propose a geometrization of the non-relativistic quantum mechanics for mixed states. Our geometric approach makes use of the Uhlmann's principal fibre bundle to describe the space of mixed states and as a novelty tool, to define a dynamic-dependent  metric tensor on the principal manifold, such that the projection of the geodesic flow to the base manifold gives the temporal evolution predicted by the von Neumann equation. Using that approach we can describe every conserved quantum observable as a Killing vector field, and provide a geometric proof for the Poincare quantum recurrence in a physical system with finite energy levels.   %
\end{abstract}

\maketitle

\section{Introduction}
The geometrization of physical theories is a successful and challenging area in theoretical physics. The most well known  examples are Hamiltonian mechanics based on symplectic geometry, General Relativity based on semi-Riemannian geometry and classical Yang-Mills theory which uses fibre bundles \cite{Jost-phy}. 

Geometric ideas have also found a clear utility in non-relativistic quantum mechanics problems because quantum theory can be formulated in the language of Hamiltonian phase-space dynamics \cite{Kibble}. Hence, the quantum theory has an intrinsic mathematical structure equivalent to Hamiltonian phase-space dynamics. However, the underlying phase-space is not the same space of classical mechanics, but  the space of quantum mechanics itself, i.e., the \emph{space of pure states} or the \emph{space of mixed states}. 

Unlike General Relativity or Gauge Theory where the metric tensor or the connection are related with the physical interaction, the most usual geometric formulation of the geometry of non-relativistic quantum mechanics is not dynamic, in the sense that is insensitive to changes in the Hamiltonian of the system. Under these assumptions, that approach only makes use of the differential structure of the Hilbert space for quantum states and the Fubini-Study metric. See for example the geometric interpretation of Berry's phase \cite{Bengtsson}.

From a more dynamical point of view, A. Kryukov\cite{kryukov2} has stated that the Schr\"odinger equation \footnote{Observe that in equation \ref{Schroed} and throughout this paper we use the natural units system ($h=1$)} for a pure state $\ket{\alpha_{t}}$  (Plank's constant is set equal to $1$) 
\begin{equation}
\label{Schroed}
\frac{d }{dt}\ket{\alpha_{t}}=-{i}{ H}\ket{\alpha_{t}}\quad,
\end{equation}
can be considered as a geodesic flow in a certain Riemannian manifold with an accurate metric which depends on the Hamiltonian of the system.  

The goal  of this paper is to generalize  the work of Kryukov for mixed states. To this end, we provide an underlying differential manifold to describe mixed states and a dynamic-dependent Riemannian metric tensor to analyze their temporal evolution. 

The mixed states are characterized by  density matrices and  the equation which plays the role of the Schr\"odinger one is the von Neumann equation \cite{Sakurai}

\begin{equation}
\label{vonNeumann}
\frac{d \rho_{t}}{dt}=-{i}\left[ H,\rho_{t}\right]\quad.
\end{equation}

To obtain the underlying differential manifold following the Uhlmann's geometrization for non-relativistic quantum mechanics\cite{Uhlmann86,Uhlmann87,Uhlmann89,Uhlmann91}, we make use of a principal fibre bundle such that its base manifold is the space of mixed states. Finally, to provide the Riemannian metric we choose an appropriate metric in the principal bundle, in such a way that the projection of the geodesic flow in the principal manifold to the base manifold  is just the temporal evolution given by  the von Neumann equation.

Among the geometric properties that are observed due to the movement of this geodesic flow, in this paper we analyze the phase volume conservation according to the Liouville Theorem. That allow us to show  a geometric proof of  the Poincare recurrence theorem relating it with the recurrence principle for physical systems with discrete energy levels. Let us emphasize that our geometric proof for the quantum Poincare recurrence is closer to the classical mechanics proof\cite{Arnold} (that also uses the conservation of the volume in the phase-space evolution) than the previous given in the quantum setting\cite{Bocchieri,Schulman,Percival}.   

\section{Density matrices space as a base of a principal fibre bundle}

The most general state, the so-called \emph{mixed state}, is represented by a \emph{density operator}  in the Hilbert space $\mathcal{H}$. In this paper we always assume that dim$(\mathcal{H})=n<\infty$, being $\mathcal{H}$ a vector space on the complex field ($\mathcal{H}=\mathbb{C}^n$). The density operator $\rho$ is in fact a \emph{density matrix}. Recall that a density matrix is a complex matrix $\rho$  that  satisfies the following properties:
\begin{enumerate}
\item $\rho$ is a hermitian matrix, i.e, the matrix coincides with its conjugate transpose matrix: $\rho=\rho^\dag$.
\item $\rho$ is positive $\rho \geq 0$. It means that any eigenvalue of $A$ is non-negative.
\item $\rho$ is normalized by the trace $\tr (\rho)=1$.
\end{enumerate}

Let us denote by $\mathcal{P}$ the space of mixed quantum states.
Note that the space of pure states $\mathcal{P}(\mathcal{H})$ is just
$$
\mathcal{P}(\mathcal{H})=\left\{ \rho  \in \mathcal{P} \, \vert\, \rho^2=\rho\right\}\quad. 
$$

Recall that the space of quantum pure states has an elegant interpretation as a $U(1)$-fibre bundle $\mathcal{S}(\mathcal{H})\to \mathcal{P}(\mathcal{H})$. Following Uhlmann \cite{Uhlmann86,Uhlmann87,Uhlmann89,Uhlmann91} and Bengtsson  and Chru\`sci\`nski books \cite{Bengtsson,Chruscinski}, we can use a similar argument to the case of quantum pure states. The key idea of Uhlmann's approach is to lift the system density operator $\rho$, acting on the Hilbert space $\mathcal{H}$, to an extended Hilbert space 
$$
\mathcal{H}^{\text{ext}}:=\mathcal{H} \otimes \mathcal{H}\quad.
$$ 

In quantum information theory \cite{Nielsen}, the procedure of extension, $\mathcal{H} \to \mathcal{H}^{\text{ext}}$ is known as attaching an ancilla living in $\mathcal{H}$. Obviously, the space of squared matrices $\Mnn$ ($n$ rows, $n$ columns) over $\mathbb{C}$ (that is a $2n^2$ real dimensional manifold) can be identified with $\mathcal{H}^{\text{ext}}$
$$
\Mnn \cong \mathcal{H}^{\text{ext}}\quad.
$$ 

Since $\tr(WW^\dag)$ is a smooth function in the space of squared matrices, by the Regular Level Set Theorem \cite{Lee-man}, the set
\begin{equation}
\mathcal{S}_0:=\left\{ W \in \Mnn \, : \, \tr (WW^\dag)=1 \right\}\quad,
\end{equation}
is a smooth manifold of $\Mnn$. Actually, it is not hard to see that $\mathcal{S}_0$ is diffeomorphic to the sphere $\mathbb{S}^{2n^2-1}$. If $\rho$ is a mixed state in $\mathcal{P}$, we shall denote an element $W\in \mathcal{S}_0$ a \emph{purification} of $\rho$ if
\begin{equation}
\label{2_2}
\begin{aligned}
&\rho=WW^\dag\quad,
\end{aligned}
\end{equation}
therefore, we get the space of density matrices $\mathcal{P}$ by the projection $\pi: \mathcal{S}_0\to \mathcal{P} $, where the projection is given by
\begin{equation}
\pi(W)=WW^\dag\quad.
\end{equation}

Observe that, if $u$ is an unitary matrix (i.e, $uu^\dag=u^\dag u=\mathbb{I}_n$) then
\begin{equation}
\pi(Wu)=\pi(W)\quad.
\end{equation}

Moreover, to fix notation recall that the Lie group $U(n)$ is a \emph{Lie transformation group}\cite{Koba1} acting on $\mathcal{S}_0$ on the right. In general, a principal fibre bundle\cite{Koba1} will be denoted by $P(M,G,\pi)$, being  $P$ the \emph{total space}, $M$ the \emph{base space}, $G$ the \emph{structure group} and $\pi$ the \emph{projection}. For each $x \in M$, $\pi ^{-1}(x)$ is a closed submanifold of $P$, called the \emph{fibre} over $x$. If $p$ is a point of $\pi^{-1}(x)$, then $\pi^{-1}(x)$ is the set of points $\left\{pa, a\in G\right\}$, and it is called fibre through $p$. 

At this point, an important question to answer,  is if $\mathcal{S}_0(\mathcal{P},U(n),\pi)$ is a principal fibre bundle over the base manifold $\mathcal{P}$ of density matrices. Unfortunately the answer is no, because $U(n)$ does not act freely on $\mathcal{S}_0$. In general, $Wu=W$ for $W\in \mathcal{S}_0$ and $u\in U(n)$ do not imply that $u=\mathbb{I}_n$, but observe that if $\det (W)\neq 0$ ( i.e, $W$ is an invertible matrix ) $U(n)$ would act freely on our space. This should be the way to describe the space of density matrices. Instead of starting with $\Mnn$, we start with the subset of invertible matrices. That subset has the differentiable structure of the Lie group $\text{GL}(n,\mathbb{C})$. Then, we build a submanifold $\mathcal{S}$ of $\text{GL}(n,\mathbb{C})$ given by
\begin{equation}
\mathcal{S}:=\left\{W\in \text{GL}(n,\mathbb{C})\,; \,\tr(WW^\dag)=1\right\}\quad.
\end{equation}

Finally, we obtain the base manifold $\mathcal{P}^+$ using the projection $\pi: \mathcal{S}\to \mathcal{P}^+$ given by
\begin{equation}
\pi(W)=WW^\dag\quad,
\end{equation}
and therefore, $\mathcal{S}(\mathcal{P}^+,U(n),\pi)$ becomes a principal fibre bundle. Observe that 
$$
\mathcal{P}^+=\left\{ \rho \in \mathcal{P} \, \vert \, \rho>0\right\}\quad,
$$
contains only strictly positive (or faithful) density operators. But $\mathcal{P}$ can be recovered from $\mathcal{P}^+$ by continuity arguments\cite{Bengtsson}.
  
In short, we describe the geometry of density matrices as a base manifold of a principal fibre bundle consisting of  a  submanifold $\mathcal{S}$ of the Lie group $GL(n,\mathbb{C})$ diffeomorphic to the sphere $\mathbb{S}^{2n^2-1}$ as a total space and the Lie group $U(n)$ as structure group.

Since $\mathcal{S}(\mathcal{P}^+,U(n),\pi)$ admits a global section \footnote{The map $\tau$ is well defined because a positive operator admits a unique positive square root. It is a section because $\pi(\tau(\rho))=(\sqrt(\rho))^2=\rho$.} $\tau: \mathcal{P}^+\to \mathcal{S}$
\begin{equation}\label{global-section}
\tau(\rho):=\sqrt{\rho} \quad,
\end{equation}
therefore $\mathcal{S}(\mathcal{P}^+,U(n),\pi)$ is a trivial bundle from a topological point of view, that means that\cite{Bengtsson,Koba1}
\begin{equation}
\mathcal{S}=\mathcal{P}^+\times U(n)\quad.
\end{equation}

\section{Hamiltonian vector field, Dynamic Riemannian metric, SHg-quantum fibre bundle and Main theorem}
In this section, we define a Riemannian metric for dynamics systems and we study how this metric acts within the tangent vector space of $\mathcal{S}$. We also discuss its relationship with other metrics such as the Bures metric or the metric proposed by Kryukov \cite{kryukov2}.

\subsection{Hamiltonian vector field, dynamic metric and its relation with other metrics}
In order to provide explicit expressions for tangent vectors to $\mathcal{S}$ and the metric tensor,  we identify the tangent space $T_W\Mnn$ with $\Mnn$ itself. Since our total space $\mathcal{S}$ is a submanifold of the manifold $\Mnn$, where each point  $W\in \mathcal{S}$ is a matrix, and the tangent space $T_W\mathcal{S}$ to $\mathcal{S}$ in the point $W$ is a subspace of the tangent space $T_W\Mnn$. We can use a matrix to describe a point $W\in \mathcal{S}$ and a matrix to describe a tangent vector $X\in T_W\mathcal{S}$ too.

First of all, note that the Hamiltonian operator $H$ induces a vector field $h: \mathcal{S}\to T\mathcal{S}$ on $\mathcal{S}$ given by
\begin{equation}
h_W:=-iHW\quad,
\end{equation} 
where $h_W$ denotes the vector field in the point $W\in \mathcal{S}$, i.e, $h_W=h(W)$. That vector field $h$ will be denoted as the \emph{Hamiltonian vector field}. 

For any point $W\in \mathcal{S}$, and any two tangent vectors $X,Y \in T_W\mathcal{S}$, we define the \emph{dynamic Riemannian metric} $g_H(X,Y)$ as 
\begin{equation}\label{3_2}
g_H(X,Y):=\frac{1}{2}\tr(X^\dag H^{-2}Y+Y^\dag H^{-2} X)\quad.
\end{equation} 

It will be denoted by $\nabla^H$ the Levi-Civita connexion (the sole metric torsion free connexion) given by $g_H$. In the definition (\ref{3_2}) we use $H^{-2}$ assuming that $H$ is an invertible matrix, but that in fact makes no restriction on the Hamiltonian of the system because we can set $H\to H+\mathbb{I}_n$ without changing the underlying physics. It is not hard to see that $g_H$ defines a positive definite inner product in each tangent space $T_W\mathcal{S}$, being therefore $g_H$ a Riemannian metric. 

With that metric tensor $g_H$ the (sub)manifold $(\mathcal{S},g_H)$ becomes a Riemannian manifold. In order to fix the notation we denote $\left\{\mathcal{S}(\mathcal{P}^+,U(n),\pi),h,g_H\right\}$ the \emph{SHg-quantum fibre bundle of dimension} $n$.

The rest of this section will examine the inherited metric in the base manifold (theorem 1) from the dynamic metric in the principal manifold and its relation with other metrics. 

The tangent space $T_W\mathcal{S}$ at the point $W\in \mathcal{S}$ can be decomposed in its horizontal $H_W$ and vertical $V_W$ subspaces:
$$
T_W\mathcal{S}=H_W\oplus V_W\quad.
$$

Observe moreover that the vertical subspaces $V_W$ are the vectors tangent to the fibres. Therefore, any vertical vector $X_V\in T_W\mathcal{S}$ can be written as
$$
X_V=W\, A\quad,
$$
where $A\in \mathfrak{u}(n)$ (i.e, $A$ is an antihermitian matrix). Note that our metric $g_H$ defines a natural connexion as follows: A tangent vector $X$ at $W$ 
is horizontal if it is orthogonal to the fibre passing through $W$, i.e., if
$$
g_H(X,Y)=0\quad,
$$
for all vertical vector $Y$ at $W$ ($Y\in V_W$). Hence $X \in T_W\mathcal{S}$ is horizontal if 
\begin{equation}
X^\dag H^{-2}W-W^\dag H^{-2}X=0\quad.
\end{equation} 

Therefore, we can define a metric $g^{\Pm}_H$ in the base manifold for any point $\rho \in \mathcal{P}^+$, given by
$$
g^{\Pm}_H(Y,Z):=g_H(Y_{\text{Hor}},Z_{\text{Hor}})\quad,
$$ 
where $Y, Z \in T_\rho \Pm$ and $Y_\text{Hor}$ (respectively $Z_\text{Hor}$) are the horizontal lift of $Y$ (respectively $Z$).

\begin{thm}The metric $g^{\Pm}_H$ in the base manifold at any point $\rho \in \mathcal{P}^+$ can be obtained as
$$
g^{\Pm}_H(Y,Z):=\frac{1}{2}\tr(H^{-1}G_YH^{-1}Z)\quad,
$$
where $G_Y$ is the unique hermitian matrix satisfying
$$
H^{-1}YH^{-1}=G_YH^{-1}\rho H^{-1} + H^{-1}\rho H^{-1}G_Y\quad.
$$
\end{thm}

Note that matrix $G_Y$ exists and is unique by the existence and uniqueness of the solution of the Sylvester equation\cite{Sylvester,Bartels}. Observe moreover that  when $H$ is the identity matrix, then
$$
g^{\Pm}_H(Y,Z):=\frac{1}{2}\tr(G_YZ)\quad, 
$$
where $G_Y$ is the (unique) solution of  
$$
Y=G_Y\rho + \rho G_Y\quad,
$$
and that is the Bures metric \cite{Dittmann}.

\begin{proof}
Let $W:\erre \to \mathcal{S}$ be a curve, such that $\dot W$ is an horizontal vector, then
$$
(\dot W)^\dag H^{-2} W=W^\dag H^{-2}\dot W\quad.
$$

Let us define $A=H^{-1}W$, thus 
$$
\dot A^\dag A=A^\dag \dot A\quad.
$$

It is easy to see that the latter condition is fulfilled if
\begin{equation}
\dot A= G A\quad,
\end{equation}
where $G$ is an Hermitian matrix. Therefore
\[
\dot W=HGH^{-1}W\quad.
\]

Hence applying equation \ref{2_2}
\begin{equation}
\label{3_5}
\pi_*(\dot W)=\dot WW^\dag+W\dot W^\dag=HGH^{-1}\rho+\rho H^{-1}GH\quad.
\end{equation}

Suppose that 
$$
\begin{aligned}
&\pi_*(\dot W)=Y\quad \pi_*(\dot V)=Z\\
&W(0)W(0)^\dag=V(0)V(0)^\dag=\rho\quad,
\end{aligned}
$$
then
\begin{equation}
\label{3_6}
\begin{aligned}
g_H^{\Pm}(Y,Z)=&g_H(\dot W,\dot V)=\frac{1}{2}\tr(\dot W^\dag H^{-2} \dot V+\dot V^\dag H^{-2} \dot W)\\
=&\frac{1}{2}\tr(H^{-1}G_YG_ZH^{-1}\rho + H^{-1}G_ZG_YH^{-1}\rho)\quad,
\end{aligned}
\end{equation}
where 
\[
\dot W=HG_YH^{-1}W\quad \dot V=HG_ZH^{-1}V\quad.
\]

Applying equation (\ref{3_5}) in $\pi_*(\dot V)$
$$
Z=HG_ZH^{-1}\rho+\rho H^{-1}G_ZH\quad.
$$

Using the above expression (\ref{3_6}) the theorem follows.
\end{proof}

In the case of pure states, our Hilbert space is $\mathbb{C}^n$ and the tangent space will be  $\mathbb{C}^n$ too. Following Kryukov work\cite{kryukov2}, we can define a metric $g_K(X,Y)$ for any two tangent vectors $X = (x,x^*)$, $Y= (y,y^*)$, by
$$
g_K(X,Y):=\text{Re}\left(\langle H^{-1}X,H^{-1}Y\rangle\right)\quad,
$$
where $\langle X,Y\rangle=\sum_{i=1}^nx_i y_i^*$, therefore
$$
g_K(X,Y):=\frac{1}{2}\left(\langle H^{-1} X,H^{-1}Y\rangle+\langle H^{-1}Y,H^{-1}X\rangle\right)=\frac{1}{2}\tr(X^\dag H^{-2}Y+Y^\dag H^{-2} X)\quad.
$$

When $H$ is the identity, we recover the Fubini-Study metric.

\subsection{Geometric structure of the SHg-quantum fibre bundle}
As we have previously seen in the SHg-quantum fibre bundle $\left\{\mathcal{S}(\mathcal{P}^+,U(n),\pi),h,g_H\right\}$ of dimension $n$, $\mathcal{S}(\mathcal{P}^+,U(n),\pi)$ is a principal (and trivial) fibre bundle, $\mathcal{S}$ is diffeomophic to the sphere of dimension $2n^2-1$, $h$ is a vector field on $\mathcal{S}$, and $(\mathcal{S},g_H)$ is a Riemannian manifold. But the SHg-quantum fibre bundle has more geometric properties :  

\begin{thm}[Main theorem]
Let $\left\{\mathcal{S}(\mathcal{P}^+,U(n),\pi),h,g_H\right\}$ be a SHg-quantum fibre bundle of dimension $n$.  Then:

\begin{enumerate}
\item $h$ is a Killing vector field of $(\mathcal{S},g_H)$.
\item The integral curves $\gamma: I\subset \erre\to \mathcal{S}$  of $h$ are geodesics of $(\mathcal{S},g_H)$.
 \item The projection on the base manifold $\Pm$ of the geodesic $\gamma$ satisfies the von Neumann equation
\begin{equation}
\frac{d }{dt}\pi\circ \gamma=-{i}\left[ H,\pi\circ \gamma\right]\quad.
\end{equation}
\end{enumerate} 
\end{thm}
\begin{proof}

Condition (1):
In order to proof that $h$ is a Killing vector field, we only have to show that the flow $\varphi_t:\mathcal{S}\to \mathcal{S}$ given by

\begin{equation}
\begin{cases}
\varphi_0(W)=W, \textnormal{where } W \in S\\
\frac{d}{dt}\varphi_t(W)\vert_{t=0}=h_W\quad,
\end{cases}
\end{equation}
is an isometry, i.e, for any $X,Y \in T_W\mathcal{S}$
\begin{equation}
g_H({\varphi_t}_*(X),{\varphi_t}_*(Y))=g_H(X,Y)\quad.
\end{equation}

Note that 
\begin{equation}
{\varphi_t}_*(X)=e^{-iHt}X\quad,
\end{equation}
and
\begin{equation}\label{isometry}
\begin{aligned}
g_H({\varphi_t}_*(X),{\varphi_t}_*(Y))=& g_H(e^{-iHt}X,e^{-iHt}Y)\\
=&\frac{1}{2}\tr\left((e^{-iHt}X)^\dag H^{-2}e^{-iHt}Y+(e^{-iHt}Y)^\dag H^{-2}e^{-iHt}X\right)\\
=&\frac{1}{2}\tr\left(X^\dag e^{iHt}H^{-2}e^{-iHt}Y+Y^\dag e^{iHt}H^{-2}e^{-iHt}X\right)\\
=&\frac{1}{2}\tr\left(X^\dag H^{-2} Y+Y^\dag H^{-2}X\right)=g_H(X,Y)\quad.
\end{aligned}
\end{equation}

Conditions (2) and (3):
First of all observe that if $\gamma$ is the integral curve of the vector field $h$, i.e, 
\begin{equation}
\dot \gamma=h_{\gamma}=-iH\gamma \quad.
\end{equation}

The projection of $\gamma$ satisfies
\begin{equation}
\begin{aligned}
\frac{d}{dt}\pi(\gamma(t))&=\frac{d}{dt}\left(\gamma(t)\gamma^\dag(t)\right)=\dot\gamma(t)\gamma^\dag(t)+\gamma(t)\dot\gamma^\dag(t)=\dot\gamma(t)\gamma^\dag(t)+\gamma(t)(\dot\gamma(t))^\dag\\
&=-iH\gamma \gamma^\dag(t)+\gamma(t)(-iH\gamma)^\dag=-i[H,\pi(\gamma(t))]\quad.
\end{aligned}
\end{equation}

Hence, the projection of the integral curves of the vector field $h$  satisfies the von Neumann equation. So all we have to prove is that curves are actually geodesic curves
\begin{equation}
\nabla_{h_\gamma}^Hh_\gamma=0\quad.
\end{equation}

Since $h$ is a Killing vector field, we only have to proof that  $h$ is a unitary vector field (due any unitary Killing vector field is a geodesic). Namely, the equality
\begin{equation}
\begin{aligned}
g_H(h_\gamma,h_\gamma)=&g_H(-iH\gamma,-iH\gamma)=\tr\left( (-iH\gamma)^\dag H^{-2}(-iH\gamma)\right)\\
=& \tr\left( \gamma^\dag HH^{-2}H\gamma\right)=\tr\left( \gamma^\dag \gamma\right)=1\quad.
\end{aligned}
\end{equation}

Finally, since $\mathcal{H}$ is a unitary Killing vector field and the integral curves of any Killing vector field of constant length is a geodesic (see appendix theorem \ref{unitary-geodesic}), the integral curves of $\mathcal{H}$ are geodesics.
\end{proof}

\begin{remark}
Let us emphasize that for any hermitian matrix $A=A^\dag$, we can build the vector field $\mathcal{A}$ on $\mathcal{S}$ given by $-iAW$ for any $W\in \mathcal{S}$. It is easy to check as done in equation \ref{isometry} that if $[H,A]=0$, $\mathcal{A}$ is a Killing vector field. Therefore, the set of operators compatible with the Hamiltonian are related to the set of isometries  of $(\mathcal{S},g_H)$, and we can identify any conserved quantum observable with a Killing vector field.
\end{remark}

\section{Geometric approach to quantum Poincare recurrence}
As we know from the main theorem, $h$ is a Killing vector field on the principal manifold $(\mathcal{S},g_H)$ endowed with the dynamic metric $g_H$. Then, the transformations given by the $1-$parametric subgroup $\varphi_t:S\to S$ of integral curves of $h$ are distance-preserving and volume-preserving (see appendix theorem \ref{killing-isometry}). These two facts have the following consequences

\begin{thm}[Insensitivity to Initial Conditions Theorem]
Let $\left\{\mathcal{S}(\Pm,U(n),\pi),h,g_H\right\}$ be a SHg-quantum principal bundle of dimension $n$. Then, for any two points  $ W,V \in  \mathcal{S}$
\begin{equation}
\text{dist} (\varphi_t(W),\varphi_t(V))=\text{dist}(W,V)\quad,
\end{equation}
being the $\varphi_t$ the $1-$parametric subgroup of transformations  given by the integral curves of the Killing field $h$.
\end{thm}

\begin{figure}[h]
\centerline{
\includegraphics[width=.45\textwidth{}]{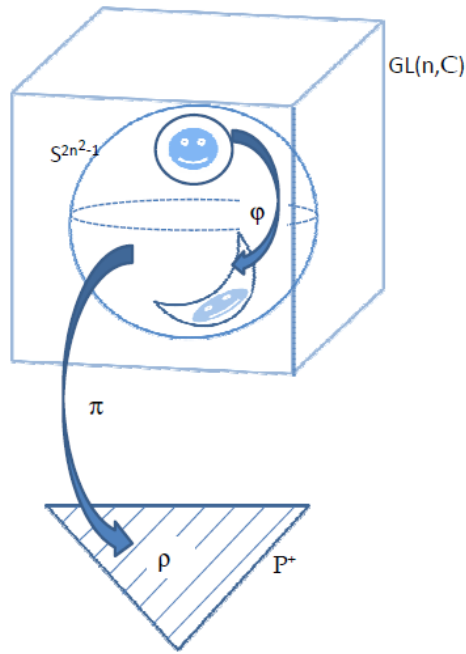}}
\caption{\small{Since the Hamiltonian vector field is a Killing vector field, its flow $\varphi$ preserves the  volume  (Liouville theorem) in the sphere $\mathcal{S}$, where each point can be projected into the space of density matrices $\Pm$.} }
\label{figure1}
\end{figure}

The classical Liouville theorem\cite{Arnold} states  that the natural volume form on a symplectic manifold is invariant under the Hamiltonian flows.  In our case, we have the $1-$parametric subgroup of transformations $\varphi_t: \mathcal{S}\to \mathcal{S}$ given by the integral curves of the Killing vector field $h$ and we can set

\begin{thm}[Liouville Type Theorem]
Let $\left\{\Pm,U(n),\pi),h,g_H\right\}$ be a SHg-quantum principal bundle of dimension $n$. Then for any domain $\Omega\subset \mathcal{S}$
\begin{equation}
\Vol(\varphi_t(\Omega))=\Vol(\Omega)\quad,
\end{equation}
being the $\varphi_t$ the $1-$parametric subgroup of transformations  given by the integral curves of the Killing vector field $h$.
\end{thm}

Using the above theorem,  we can therefore state a similar theorem to the Poincare recurrence theorem\cite{Arnold}. 

\begin{thm}[Poincare Type Theorem]
Let $\left\{\mathcal{S}(\Pm,U(n),\pi),h,g_H\right\}$ be a SHg-quantum principal bundle of dimension $n$. For any domain $\Omega \subset \mathcal{S}$ and any time period $T\in \erre^+$ there exist a point $x\in \Omega$ and a positive integer $k>0$ such that 
\begin{equation}
\varphi_{kT}(x)\in \Omega\quad,
\end{equation}
being $\varphi_t: \mathcal{S}\to \mathcal{S}$ the $1-$parametric subgroup of transformations given by the integral curves of the Killing field $h$.
\end{thm}

\begin{proof}
Consider the following sequence of domains 
$$
\Omega, \varphi_T(\Omega),\varphi_{2T}(\Omega), \cdots, \varphi_{kT}(\Omega), \cdots
$$

All domain in the sequence belongs to the same volume $\Vol(\Omega)$. If the above domains never intersect $\mathcal{S}$, an infinite volume would obtain, but $\mathcal{S}$ is compact, so $\Vol(\mathcal{S})<\infty$. Then, there exist $l\geq 0$ and $m>l$ such that
\begin{equation}
\varphi_{lT}(\Omega)\cap \varphi_{mT}(\Omega)\neq \emptyset\quad,
\end{equation}
so
\begin{equation}
\Omega\cap \varphi_{(m-l)T}(\Omega)\neq \emptyset\quad.
\end{equation}

Setting  $k=m-l$ the theorem is proven.
\end{proof}
Joining the above theorem with the Insensitivity to the Initial Conditions we get

\begin{thm}[Strong Poincare Type Theorem]
Let $\left\{\mathcal{S}(\Pm,U(n),\pi),h,g_H\right\}$ be a SHg-quantum principal bundle of dimension $n$. Then, for any point $W \in \mathcal{S}$, any $\epsilon>0$ and any $T\in \erre^+$, there exist a positive integer $k>0$ such that 
\begin{equation}
\text{dist}(W,\varphi_{kT}(W))< \epsilon\quad,
\end{equation}
being $\varphi_t: \mathcal{S}\to \mathcal{S}$ the $1-$parametric subgroup of transformations given by the integral curves of the Killing vector field $h$.
\end{thm}
\begin{proof}
Let us consider the domain
\begin{equation}
B_{\frac{\epsilon}{2}}(W)=\left\{V\in \mathcal{S}\, :\, \text{dist }(W,V)<\frac{\epsilon}{2}\right\}\quad.
\end{equation}

Applying now the Poincare type theorem there must exist $W_{0} \in B_{\frac{\epsilon}{2}}(W)$ and $k>0$ such that
\begin{equation}
\varphi_{kT}(W_{0})\in B_{\frac{\epsilon}{2}}(W)\quad.
\end{equation} 

So, 
\begin{equation}
\text{dist }(W,\varphi_{kT}(W_{0}))<\frac{\epsilon}{2}\quad.
\end{equation} 

But, by the Insensitivity to Initial Conditions Theorem 
\begin{equation}
\text{dist }(\varphi_{kT}(W),\varphi_{kT}(W_{0}))=\text{dist }(W,W_{0})<\frac{\epsilon}{2}\quad.
\end{equation} 

Therefore, applying the triangular inequality
\begin{equation}
\text{dist }(W,\varphi_{kT}(W))\leq \text{dist }(W,\varphi_{kT}(W_{0}))+ \text{dist }(\varphi_{kT}(W_{0}),\varphi_{kT}(W))<\epsilon\quad.
\end{equation} 
\end{proof}

\subsection{Physical systems with discrete energy eigenvalues}
Using previously stated theorems we can give an alternative proof and more geometric sense of well-known\cite{Bocchieri,Schulman,Percival} principle of recurrence for physical systems with discrete energy eigenvalues.

Thus, defining the length $\Vert A \Vert $ of a matrix $A$ as follows\cite{Percival}
$$
\Vert A\Vert= \sqrt {\tr(A^\dag A)}\quad.
$$ 

Then

\begin{thm}
Let $\rho$ be a mixed state of a quantum system with discrete energy spectrum. Then, $\rho$ is almost periodic. Namely, for an arbitrarily small positive error $\epsilon$ the inequality
\begin{equation}
\Vert \rho(t+T)-\rho(t)\Vert < \epsilon  \text{ for all } t
\end{equation}
is satisfied by infinitely many values of $T$, these values being spread over the whole range $-\infty$ to $\infty$ so as not to leave arbitrarily long empty intervals.
\end{thm}
\begin{proof}
Let $\rho(t)$ be the density matrix of a system with a discrete set of stationary states, labeled $n=0,1,2,\cdots,$ with energies $E_n$, some of which may be equal if there are degeneracies. In energy representation the matrix elements are
\begin{equation}
\rho_{nn'}(t)=\langle n\vert \rho(t)\vert n'\rangle\quad.
\end{equation}

Let $T_n=\vert n\rangle\langle n\vert$ be the projection operator onto the $n$th stationary state, then
\begin{equation}
\rho^{nn'}(t)=T_n\rho(t)T_{n'}\quad,
\end{equation}
is the matrix which energy representation has only one nonzero element, equal to $\rho_{nn'}(t)$ and in the location $(n,n')$. These matrices are orthogonal in density space
\begin{equation}
\left(\rho^{nn'}(t),\rho^{n''n'''}(t)\right)=\delta_{nn''}\delta_{n'n'''}\vert\rho_{nn'}(t)\vert^2\quad,
\end{equation}
and
\begin{equation}
\begin{aligned}
\rho(t)&=\sum_{n=0}^{\infty}\sum_{n'=0}^{\infty}\rho^{nn'}(t)\\
&=\sum_{n=0}^{\infty}\sum_{n'=0}^{\infty}\rho^{nn'}(0)e^{i\omega_{nn'}t}\quad,
\end{aligned}
\end{equation}
where $\omega_{nn'}=(E_{n'}-E_n)$.
Now, consider the finite sum
\begin{equation}
\sigma^{NN'}(t)=\sum_{n=0}^N \sum_{n'=0}^{N'}\rho^{nn'}(t)\quad,
\end{equation}
as an approximation to $\rho(t)$. The square of the error is
\begin{equation}
\begin{aligned}
\Vert \rho(t)-\sigma^{NN'}(t)\Vert^2&=\Vert \sum_{n=N+1}^\infty \sum_{n'=N'+1}^{\infty}\rho^{nn'}(t)\Vert^2\\
&=\sum_{n=N+1}^\infty \sum_{n'=N'+1}^{\infty}\Vert \rho^{nn'}(t)\Vert^2\\
&=\sum_{n=N+1}^\infty \sum_{n'=N'+1}^{\infty}\Vert \rho^{nn'}(0)\Vert^2\quad.
\end{aligned}
\end{equation}

The second equality follows from the orthogonality of the $\rho^{nn'}$. Since the error is independent of the time, $\sigma^{NN'}(t)$ converges uniformly to $\rho(t)$ (in the $\Vert\, \Vert$-norm sense). So, $\rho(t)$ can be approximated by $\sigma^{NN'}(t)$. $\sigma^{NN'}(t)$ is a discrete density with finite energy levels, $\sigma^{NN'}(t)\in \Pm$, and the set
\begin{equation}
B_{\epsilon}^{\Pm}(\sigma^{NN'}):=\{ \rho \in \Pm \, : \,  \Vert \rho-\sigma^{NN'}(t)\Vert < \epsilon\}\quad,
\end{equation}
is an open precompact set in $\Pm$. Using the global section given in equation (\ref{global-section}), $\tau(B_\epsilon)$ will be an open precompact set of $\mathcal{S}$. But applying the Strong Poincare Type Theorem for any time period $T>0$ there exists $k>0$ such that
\begin{equation}
\text{dist}\left(\tau(\sigma^{NN'}(t)),\varphi_{kT}(\tau(\sigma^{NN'}(t)))\right)=\text{dist}\left(\tau(\sigma^{NN'}(t)),\tau(\sigma^{NN'}(t+kT))\right)<\varepsilon\quad,
\end{equation} 
for any $\varepsilon >0$. Namely,
\begin{equation}\label{pertany}
\tau(\sigma^{NN'}(t+kT))\in B_{\varepsilon}^{\mathcal{S}}(\tau(\sigma^{NN'}(t)))\quad,
\end{equation}
being $B_{\varepsilon}^{\mathcal{S}}(\tau(\sigma^{NN'}(t)))$ the geodesic ball in $\mathcal{S}$ centered at $\tau(\sigma^{NN'}(t))$ of radius $\varepsilon$. Now choosing $\varepsilon$ small enough 
\begin{equation}
B_{\varepsilon}^{\mathcal{S}}(\tau(\sigma^{NN'}(t)))\subset \tau(B^{\Pm}_\epsilon(\sigma^{NN'}(t)))\quad.
\end{equation}

Therefore by (\ref{pertany})
\begin{equation}
\tau(\sigma^{NN'}(t+kT))\in \tau(B^{\Pm}_\epsilon(\sigma^{NN'}(t)))\quad.
\end{equation}
 
Projecting to the base manifold
\begin{equation}
\sigma^{NN'}(t+kT)\in B^{\Pm}_\epsilon(\sigma^{NN'}(t))\quad.
\end{equation}

By definition of $B^{\Pm}_\epsilon(\sigma^{NN'}(t))$
\begin{equation}
\Vert \sigma^{NN'}(t+kT)-\sigma^{NN'}(t)\Vert < \epsilon\quad.
\end{equation}

And the theorem is proven.
\end{proof}
\section{Appendix}
In this section we recall several well known results about Killing vector fields on Riemannian manifolds(for a more detailed approximation see O'Neill\cite{Oneill}). 

\begin{thm}\label{unitary-geodesic}(see also \cite{Beres})Let $(M,g)$ be Riemannian manifold, then  any integral curve $\gamma: I\subset \erre\to M$ of a Killing vector field $X$ of constant length $\sqrt{g(X,X)}$ is a geodesic on $M$. 
\end{thm}
\begin{proof}
Here, we need
\begin{equation}
\nabla_{\dot \gamma}\dot \gamma=\nabla_XX=0\quad,
\end{equation}
but, since $X$ is a Killing vector field, the Lie derivative of the metric is zero $L_Xg=0$ and (see O'Neill\cite{Oneill}, proposition 25) $\nabla X$ is skew-adjoint relative to $g$, then
\begin{equation}
g(\nabla_XX,W)+g(\nabla_WX,X)=0\quad,
\end{equation}
for any $X\in T\mathcal{S}$. Therefore
\begin{equation}
0=g(\nabla_XX,W)+1/2 W(g(X,X))=g(\nabla_XX,W)\quad,
\end{equation}
then $\nabla_XX=0$.
\end{proof}

\begin{thm}\label{killing-isometry}
Let $(M,g)$ be a Riemannian manifold, let $X$ be a Killing vector field on $M$, and denote by $\varphi_t: M\to M$ the $1-$parametric subgroup of transformations given by $X$ (i.e, $\varphi_0(p)=p$, $\frac{d}{dt}\varphi_t(p)\vert_{t=0}=X_p$), then
\begin{enumerate}
\item Given any two points $p,q \in M$, $\text{dist}(p,q)=\text{dist}(\varphi_t(p),\varphi_t(q))$.
\item Given any domain $\Omega\subset M$, $\Vol(\varphi_t(\Omega))=\Vol(\Omega)$.
\end{enumerate}
\end{thm}
\begin{proof}
Let $\varphi_t(\Omega)$ be the flow of the domain $\Omega$. Allow us denote 
\begin{equation}
V(t):=Vol(\varphi_t(\Omega))\quad.
\end{equation}

Then, the divergence is just (see Chavel\cite{Chavel2}) 
\begin{equation}
V'(0)=\int_\Omega \text{div }\mathcal{H} \, d\mu_{g_H}\quad ,
\end{equation}
where $d\mu_{g_H}$ denotes the Riemannian density measure.

The divergence is defined as\cite{DoCarmo2}
\begin{equation}
\text{div }\mathcal{H}=\tr(Y\to \nabla^H_YX)\quad.
\end{equation}

Given an orthonormal  base $\{E_i\}_{i=1}^{2n^2-1}$ in $T_W\mathcal{S}$
\begin{equation}
\nabla^H_Y\mathcal{H}=\sum_i Y^i\nabla^H_{E_i}\mathcal{H}=\sum_{i,j} Y^ig_H(\nabla^H_{E_i}\mathcal{H},E_j)E_j\quad ,
\end{equation}
where $Y^i:=g_H(Y,E_i)$. Therefore
\begin{equation}
\text{div }\mathcal{H}=\sum_ig_H(\nabla^H_{E_i}\mathcal{H},E_j)\quad.
\end{equation}

But since $\mathcal{H}$ is a Killing vector field $\nabla^H \mathcal{H}$  is skew-adjoint relative to $g_H$ (see O'Neill\cite{Oneill}, proposition 25 again), then 
\begin{equation}
\text{div }\mathcal{H}=0\quad .
\end{equation}

And the theorem follows.
\end{proof}

\def\cprime{$'$} 

\end{document}